\documentclass[12pt]{article}

\usepackage{latexsym,graphics,epsfig,amsmath,amssymb,tabularx,booktabs,apacite,natbib}
\usepackage{paralist, multirow,color, xcolor}

\usepackage{enumerate}
\usepackage{amsthm}
\usepackage{url}
\usepackage[mathscr]{eucal}

\usepackage[T1]{fontenc}

\theoremstyle{plain}
\newtheorem{proposition}{Proposition}
\newtheorem{theorem}{Theorem}

\newtheorem{example}{Example}

\def\bd#1{\mbox{\boldmath $#1$}}
\newcommand{\tr}{^{\prime}}

\def\cg#1{\ensuremath{\mathcal{#1}}}

\newenvironment{keywords}{%
\begin{changemargin}{1cm}{1cm}
\noindent{\bf Keywords:}
}
{\end{changemargin} }

\newenvironment{changemargin}[2]{%
\begin{list}{}{%
\setlength{\topsep}{0pt}%
\setlength{\leftmargin}{#1}%
\setlength{\rightmargin}{#2}%
\setlength{\listparindent}{\parindent}%
\setlength{\itemindent}{\parindent}%
\setlength{\parsep}{\parskip}%
}%
\item[]}{\end{list}}

\newcommand*\xbar[1]{%
  \hbox{%
    \vbox{%
      \hrule height 0.5pt 
      \kern0.5ex
      \hbox{%
        \kern-0.25em
        \ensuremath{#1}%
        \kern-0.1em
      }%
    }%
  }%
}

\begin{document}

\title{On the role of the overall effect in exponential families}




\author{Anna Klimova \\
{\small{National Center for Tumor Diseases (NCT), Partner Site Dresden, and}}\\
{\small{Institute for  Medical Informatics and Biometry,}}\\ 
{\small{Technical University, Dresden, Germany} }\\
{\small \texttt{anna.klimova@nct-dresden.de}}\\
{}\\
\and 
Tam\'{a}s Rudas \\
{\small{Center for Social Sciences, Hungarian Academy of Sciences, and}}\\
{\small{Department of Statistics, E\"{o}tv\"{o}s Lor\'{a}nd University, Budapest, Hungary}}\\
{\small \texttt{rudas@tarki.hu}}\\
}

 \date{}
   
\maketitle

\begin{abstract}
\vspace{2mm}

\noindent Exponential families of discrete probability distributions when the normalizing constant (or overall effect) is added or removed are compared in this paper. The latter setup, in which the exponential family is curved, is particularly relevant when the sample space is an incomplete Cartesian product or when it is very large, so that the computational burden is significant. The lack or presence of the overall effect has a fundamental impact on the properties of the exponential family. When the overall effect is added, the family becomes the smallest regular exponential family containing the curved one. The procedure is related to the homogenization of an inhomogeneous variety discussed in algebraic geometry, of which a statistical interpretation is given as an augmentation of the sample space. The changes in the kernel basis representation when the overall effect is included or removed are derived. The geometry of maximum likelihood estimates, also allowing zero observed frequencies, is described with and without the overall effect, and various algorithms are compared. The importance of the results is illustrated by an example from cell biology, showing that routinely including the overall effect leads to estimates which are not in the model intended by the researchers. 
\end{abstract}




\vspace{1cm}



%
\section{Introduction}

This paper deals with exponential families of probability distributions over discrete sample spaces. When defining such families, usually, a normalizing constant, which of course, is constant over the sample space but not over the family, is included. The presence of the normalizing constant implies that the parameter space may be an open set, which, in turn, is necessary for asymptotic normality of estimates and for the applicability of standard testing procedures. The normalizing constant, from an applied perspective, may be interpreted as a baseline or common effect, present everywhere on the sample space and is, therefore, also called the overall effect. The focus of the present work is to better understand the implications of having or not having an overall effect in such families, in particular how adding or removing the overall effect affects the properties of discrete exponential families.

Motivated by a number of important applications, \cite{KRD11, KRbm,KRextended} developed the theory of relational models, which generalize discrete exponential families, also called log-linear models, to situations when the sample space is not necessarily a full Cartesian product, the statistics defining the exponential family are not necessarily indicators of cylinder sets, and the overall effect is not necessarily present. Exponential families without the overall effect are particularly relevant, sometimes necessary, when the sample space is a proper subset of a Cartesian product.  Several real examples, when certain combinations of the characteristics were either not possible logically or were left out from the design of the experiment were discussed in \cite{KRD11}. A real problem of this structure from cell biology is analyzed in this paper, too. When the overall effect is not present, the standard normalization procedure to obtain probability distributions cannot be applied, because the family is curved \cite{KRD11}. When, in spite of this, the standard normalization procedure is applied, as was done in this analysis, the resulting estimates do not possess the fundamental model properties. 

The standardization of the estimates in exponential families is also an issue, when the size of the problem is very large and the computational burden is significant. Some Neural probabilistic language models are  relational models. Due to the high-dimensional sample space, the evaluation of the partition function, which is needed for  normalization, may be intractable. Some of the methods of parameter estimation under such models are based on the removal of the partition function, that is,  the removal of the overall effect from the model and performing model training using the models without the overall effect. Approximations of estimates with and without the overall effect were studied, for example, by \cite{MnihTeh2012} and \cite{AndreasKlein2015}, among others.  A different approach to avoiding global normalization (i.e., having an overall effect) is described in \cite{ProbGraphM}. However, the implications of the removal of the overall effect are not discussed in the existing literature. 

Another area where removing or including the overall effect is relevant, is context specific independence models, see, e.g., \cite{HosgaardCSImodels} and \cite*{NymanCSI}. When the sample space is an incomplete Cartesian product, removing the overall effect, as described in this paper, specifies different variants of conditional independence in the parts of the sample space, depending on whether or not the part is or is not affected affected by the missing cells.

While including the overall in the definition of the statistical model to be investigated is seen by many researchers as ``natural' or ``harmless'', we show in this paper that adding or removing the overall effect may dramatically change the characteristics of the exponential family, up to the point of altering the fundamental model property intended by the researcher. 

The main results of the paper include showing that allowing the overall effect expands the curved exponential family to the smallest regular exponential family which contains it.  When the overall effect is removed, the sample space may have to be reduced (if there were cells which contained the overall effect only), and the changes in the structure of the generalized odds ratios defining the model are described in both cases. In the language of algebraic geometry, the procedure of removing the overall effect is identical to the dehomogenization of the variety defining the model \citep*{Cox}. An important area of applications of the results presented here is the case when several binary features are observed, but the combination that  no feature is present is either is impossible logically or is possible but was left out from the study design. The converse of dehomogenization, that is homogenizing a variety, involves including a new variable, and it is shown that in some cases this can be identified, from a statistical perspective, with augmenting the sample space by a cell which is characterized by no feature being present. For example, the Aitchison -- Silvey  independence \citep{AitchSilvey60,KRipf1} is homogenized, through the augmentation of the sample space, into the standard independence model.

The paper is organized as follows. Section \ref{SectionDefinition}  gives a canonical definition of relational models using homogeneous, and if there is no overall effect included, one inhomogeneous generalized odds ratios, called dual representation and shows that including the overall effect is identical to omitting the inhomogeneous generalized odds ratio from it.

Section \ref{SectionInfluence} contains the result that including the overall effect expands the curved exponential family into the smallest regular one containing it. For the case of the removal of the overall effect,  the dual representation of the model is given, and the relevance of certain results in algebraic geometry to the statistical problem is discussed. In particular, the homogenization of a variety through including a new variable is identified with augmenting the sample space with a cell where no feature is present, when this is meaningful. It is proved that the homogenization of the Aitchison -- Silvey (in the sequel, AS) independence model, which is defined on sample spaces where all combinations of features, except for the ``no feature present'' combination, are possible, is the usual model of mutual independence on the full Cartesian product obtained after augmenting the sample space with the missing cell. The relationship of these results with context specific independence is also described. 

Section \ref{MLEsection} compares the maximum likelihood (ML) estimates in geometrical terms for relational models with and without the overall effect and based on the insight obtained, a modification of the algorithm proposed in \cite{KRipf1} is given. It is illustrated, that the ML estimates under two models which differ only in the lack or presence of the overall term, may be very different, up to the point of the existence or no existence of positive ML estimates, when the data contain observed zeros. However, when the MLE exists in the model containing the overall effect, it also does in the model obtained after the removal of the overall effect.

Finally, Section \ref{hemato} discusses an example of applications of relational models in cell biology. The equal loss of potential model in hematopoiesis \citep*{Perie2014}  is a relational model  without the overall effect. The published analysis of this model added the overall effect to it, to simplify calculations, and with this changed the properties of the model so that the published estimates do not fulfill the fundamental model property.

\section{A canonical form of relational models}\label{SectionDefinition}

Let $Y_1, \dots, Y_K$ be  random variables taking values in finite sets $\mathcal{Y}_1, \dots, \mathcal{Y}_K$, respectively. Let the sample space $\mathcal{I}$ be a non-empty, proper or improper, subset of $\mathcal{Y}_1 \times \dots \times \mathcal{Y}_K$, written as a sequence of length $I=|\cg I|$ in the lexicographic order. Assume that the population distribution is parameterized by cell probabilities $\boldsymbol p =(p_1, \dots, p_I)$, where $p_i \in (0,1)$ and $\sum_{i=1}^I p_{i} = 1$, and denote by $\cg{P}$ the set of all strictly positive distributions on $\cg{I}$.  For simplicity of exposition, a distribution in $\cg {P}$ will be identified with its parameter, $\boldsymbol p$, and $\cg P = \{\bd p > \bd 0: \,\, \bd 1\tr \bd p = 1\}$.  

Let $\mathbf{A}$ be a $J \times I$ matrix of full row rank with 0--1 elements and no zero columns.  
A relational model for probabilities $RM(\mathbf{A})$ generated by $\mathbf{A}$ is the subset of $\mathcal{P}$ that satisfies:
\begin{equation}\label{genll}
RM(\mathbf{A}) = \{\bd  p \in \mathcal{P}: \,\, \log \bd{p} = \mathbf{A}\tr \bd{\theta}\},
\end{equation}
where $\bd \theta$ = $(\theta_1, \dots, \theta_J)\tr$ is the vector of log-linear parameters of the model.  A dual representation of a relational model can be obtained using a matrix, $\mathbf{D}$, whose rows form a basis of $Ker(\mathbf{A})$, and thus, $\mathbf{D}\mathbf{A}\tr$ = $\mathbf O$:
\begin{equation}\label{multFam}
RM(\mathbf{A}) = \{\bd p \in \cg P: \,\, {\mathbf D} \log \bd p = \bd 0\}.
\end{equation}
The number of the degrees of freedom $K$ of the model is equal to $dim(Ker(\mathbf{A}))$.  In the sequel, $\bd d_1\tr, \bd d_2\tr, \dots, \bd d_K\tr$ denote the rows of $\mathbf D$. The dual representation can also be expressed in terms of the generalized odds ratios:
\begin{equation}\label{pDRatio}
\boldsymbol p^{\boldsymbol d_1^+}/\boldsymbol p^{\boldsymbol d_1^-} = 1, \quad 
\boldsymbol p^{\boldsymbol d_2^+} /\boldsymbol p^{\boldsymbol d_2^-} = 1, \quad
\cdots \quad
\boldsymbol p^{\boldsymbol d_K^+} /\boldsymbol p^{\boldsymbol d_K^-} = 1,
\end{equation}
or in terms of the cross-product differences:
\begin{equation}\label{pDiff}
\boldsymbol p^{\boldsymbol d_1^+} - \boldsymbol p^{\boldsymbol d_1^-} = 0, \quad 
\boldsymbol p^{\boldsymbol d_2^+} - \boldsymbol p^{\boldsymbol d_2^-} = 0, \quad
\cdots \quad
\boldsymbol p^{\boldsymbol d_K^+} - \boldsymbol p^{\boldsymbol d_K^-} = 0,
\end{equation}
where ${\boldsymbol {d^+}}$ and ${\boldsymbol {d^-}}$ stand for, respectively, the positive and negative parts of a vector $\boldsymbol d$. 
The following dual representation is invariant of the choice of the kernel basis. 

Let $\mathcal{X}_{\mathbf A}$ denote the polynomial  variety associated with $\mathbf{A}$ \citep{SturBook}:
\begin{equation}\label{variety}
\mathcal{X}_{\mathbf A} = \left\{\boldsymbol p \in \mathbb{R}^{|\mathcal{I}|}_{\geq 0}: \,\, \boldsymbol p^{\boldsymbol d^+} = \boldsymbol p^{\boldsymbol d^-}, \,\,  \forall  \boldsymbol d \in Ker(\mathbf{A}) \right \}.
\end{equation}
The relational model generated by $\mathbf{A}$ is  the following set of distributions: 
\begin{equation}\label{ExtMprob}
RM(\mathbf{A}) = \mathcal{X}_{\mathbf A} \cap int(\Delta_{I-1}),
\end{equation} 
where $int(\Delta_{I-1})$ is the interior of the $(I-1)$-dimensional simplex. 

Notice that the variety $\mathcal{X}_{\mathbf A}$ includes elements $\boldsymbol p$ with zero components as well and can be used to extend the definition of the model to allow zero probabilities. The extended relational model, $\xbar{RM}(\mathbf{A})$, is  the intersection of the variety $\mathcal{X}_{\mathbf A}$ with the probability simplex:
\begin{equation}\label{ExtMprob}
\,\xbar{RM}(\mathbf{A}) = \boldsymbol p \in \mathcal{X}_{\mathbf A} \cap \Delta_{I-1}.
\end{equation} 
See \cite{KRextended} for more detail on the extended relational models.


Let $\bd 1\tr$ = $(1, \dots,1)$ be the row of $1$'s of length $I$. If $\boldsymbol 1\tr$ does not belong to the space spanned by the rows of $\mathbf A$, the relational model  $RM(\mathbf{A})$ is said to be a model without the overall effect. Such models are specified using homogeneous and at least one non-homogeneous generalized odds ratios, and the corresponding variety $\mathcal{X}_{\mathbf A}$ is non-homogeneous \citep{KRD11}.
\begin{proposition}\label{oneORtheorem}
Let $RM(\mathbf A)$ be a model without the overall effect. There exists a kernel basis matrix $\mathbf{D}$ whose rows satisfy:
\begin{equation}\label{KernelRows}
\bd d_1\tr  \bd 1  \neq 0, \quad \bd d_2\tr \bd 1=  0, \quad \dots, \,\, \bd d_K\tr \bd 1 = 0.
\end{equation}
\end{proposition}

\begin{proof}
A relational model does not contain the overall effect if and only if it can be written using non-homogeneous (and possibly homogeneous) generalized odds ratios \citep{KRD11}. Therefore, $\mathbf{D}$ has at least one row, say $\bd d_1\tr$, that is not orthogonal to $\bd 1$: $C_1 = \bd d_1\tr \bd 1\neq 0$.

Suppose there exists another row, say $\boldsymbol d_2\tr$, that is not orthogonal to $\boldsymbol 1$ and thus $C_2 =  \boldsymbol d_2\tr\boldsymbol 1 \neq 0$. The vectors $\boldsymbol d_1$ and $\boldsymbol d_2$ are linearly independent, so are the vectors
$\boldsymbol d_1$ and $C_2\boldsymbol d_1 - C_1 \boldsymbol d_2$. Substitute the row  $\boldsymbol d_2\tr$ with the row $C_2\boldsymbol d_1\tr - C_1 \boldsymbol d_2\tr$. And so on.
\end{proof}
\noindent It is assumed in the sequel that $\boldsymbol 1\tr$ is not in the row space of $\mathbf{A}$.  Notice that, because $\mathbf{A}$ is 0--1 matrix without zero columns, this is only possible when $2 \leq J = rank(\mathbf{A})  < I-1$. Throughout the entire paper, the kernel basis matrix $\mathbf D$ is assumed to satisfy (\ref{KernelRows}), and, without loss of generality, $\bd d_1\tr\bd 1=-1$. 

Some consequences of adding the overall effect to a relational model will be investigated by comparing the properties of the relational model generated by $\mathbf{A}$ and the model generated by the matrix $\bar{\mathbf A}$ obtained by augmenting the model matrix $\mathbf{A}$ with the row $\boldsymbol 1'$: 
$$\bar{\mathbf A} = \left( \begin{array}{c} \boldsymbol 1' \\ \mathbf{A} \end{array} \right).$$
Let $RM(\bar{\mathbf{A}})$ be the relational model generated by $\bar{\mathbf{A}}$. Because $\boldsymbol 1'$ is a row of $\bar{\mathbf{A}}$, the corresponding polynomial variety $\mathcal{X}_{\bar{\mathbf A}}$ is homogeneous \citep[cf.][]{SturBook}. 


\begin{theorem}\label{OEnonhom}
The dual representation of $RM(\bar{\mathbf{A}})$  can be obtained from the dual representation of $RM({\mathbf{A}})$  by removing  the constraint specified by a non-homogeneous odds ratio from the latter. 
\end{theorem}
\begin{proof}
Write the dual representation of $RM({\mathbf{A}})$ in terms of the generalized  log odds ratios:
\begin{equation}\label{multFam1d}
\bd d_1\tr\log \bd p = 0, \quad \bd d_2\tr\log \bd p = 0, \dots, \quad \bd d_K\tr\log \bd p = 0, \quad \mbox{ for any } \,\, \boldsymbol p \in RM({\mathbf{A}}).
\end{equation}
By the previous assumption, $\bd d_1\tr\bd 1=-1$, and thus, the constraint $\bd d_1\tr\log \bd p = 0$ is specified by a non-homogeneous odds ratio.  Define $\bar{\mathbf{D}}$ as:
$$\bar{\mathbf{D}} = \left( \begin{array}{c} \boldsymbol d_2\tr \\ \vdots \\ \boldsymbol d_K\tr  \end{array} \right).$$
Because, from (\ref{KernelRows}), $\bd d_2, \dots, d_K \in Ker(\mathbf{A})$, 
$$\bar{\mathbf D} \bar{\mathbf A}\tr =   \left( \begin{array}{c} \boldsymbol d_2\tr \\ \vdots \\ \boldsymbol d_K\tr  \end{array} \right) \left( \begin{array}{cc} \boldsymbol 1 & {\mathbf A}\tr \end{array} \right) =   \left( \begin{array}{cc} \boldsymbol d_2\tr \boldsymbol 1  &  \boldsymbol d_2\tr  \mathbf{A}' \\ \vdots & \vdots \\ \boldsymbol d_K\tr \boldsymbol 1  &  \boldsymbol d_K\tr  \mathbf{A}' \end{array} \right) = \mathbf O,$$
and thus, $\bd d_2, \dots, d_K \in Ker(\bar{\mathbf{A}})$. Finally, as $rank(\bar{\mathbf{D}}) = K-1$, $\bd d_2, \dots, d_K$ is a basis of $Ker(\bar{\mathbf{A}})$, and therefore, 
\begin{equation}\label{multFam1d1}
 \bd d_2\tr\log \bd p = 0, \dots, \quad \bd d_K\tr\log \bd p = 0, \quad \mbox{ for any } \,\, \boldsymbol p \in RM(\bar{\mathbf{A}}).
\end{equation}
\end{proof}

%
\section{The influence of the overall effect on the model structure}\label{SectionInfluence}

The consequences of adding or removing the overall effect will be studied separately. The changes in the model structure after the overall effect is added  are considered first.

 Let $RM(\mathbf{A})$ be a relational model without the overall effect and $RM(\bar{\mathbf{A}})$ be the corresponding augmented model. Let $\mathbf{A} = (a_{ji})$ for $j = 1, \dots, J$, $i = 1, \dots, I$. For any $\boldsymbol p \in RM(\bar{\mathbf{A}})$:
 $$\log p_i = \theta_0 + a_{1i}\theta_1 + \dots + a_{Ji} \theta_J,$$
 where $\theta_j = \theta_j(\boldsymbol p)$, $j = 0, 1,\dots, J$, are the log-linear parameters of $\boldsymbol p$. In particular, $\theta_0(\boldsymbol p)$ is the overall effect of $\boldsymbol p$. 

 \begin{theorem}\label{newTheoremVarieties}
The augmented model, $RM(\bar{\mathbf{A}})$, is the minimal regular exponential family which contains $RM(\mathbf{A})$, and 
$$ RM(\mathbf{A}) = \{\boldsymbol p \in RM(\bar{\mathbf{A}}): \,\, \theta_0(\boldsymbol p) = 0 \}.$$
 \end{theorem}
 
 \begin{proof}

  The second claim is proved first. Denote $\mathcal{M}_{0} = \{\boldsymbol p \in RM(\bar{\mathbf{A}}): \,\, \theta_0(\boldsymbol p) = 0 \}$.  Let $\mathbf{D}$ be a kernel basis matrix of $\mathbf{A}$, having the form (\ref{KernelRows}), and notice that
 $$ \mathbf{D} \log \boldsymbol p = \left(\begin{array}{c}\bd d_1\\[3pt] \bar{\mathbf{D}}\end{array}\right)\log \boldsymbol p=\left(\begin{array}{c}\theta_0(\boldsymbol {p})\bd d_1\tr \bd 1\\[3pt] \bar{\mathbf{D}}\mathbf{A}\tr \boldsymbol \theta\end{array}\right) = \left(\begin{array}{c}- \theta_0(\boldsymbol {p})\\[3pt] \bd 0\end{array}\right).$$ 
Therefore, any $\boldsymbol p \in \mathcal{M}_0$, satisfies $ \mathbf{D} \log \boldsymbol p = \boldsymbol 0$, and thus, belongs to  $ RM(\mathbf{A})$.   On the other hand, for any $\boldsymbol p \in  RM({\mathbf{A}})$, both $\bar{\mathbf{D}}\log\boldsymbol p = \boldsymbol 0$ and  $\theta_0(\boldsymbol {p}) = 0$ must hold, which immediately implies that $\boldsymbol p \in \mathcal{M}_0$.

 The first claim is proved next.

The relational model $RM({\mathbf{A}})$ is a curved exponential family parameterized by
\begin{equation*}
\Theta = \{\bd  \theta\tr = (\theta_1, \dots, \theta_J)\tr \in \mathbb{R}^{J}: \,\, \bd 1\tr \exp\{\mathbf{A}\tr \bd{\theta}\} = 1\}.
\end{equation*}
If the overall effect is added to   $RM({\mathbf{A}})$, the parameter space gets an additional parameter:
 \begin{equation*}
\Theta_1 = \{(\theta_0, \bd \theta\tr) = (\theta_0, \theta_1, \dots, \theta_J)\tr \in \mathbb{R}^{J+1}: \,\, \theta_0 = -\log(\bd 1\tr \exp\{\mathbf{A}\tr \bd{\theta}\})\}.
\end{equation*}
Because $\Theta_1$ is the smallest open set in $\mathbb{R}^J$  that contains $\Theta$, it parameterizes the minimal regular exponential family containing $RM(\mathbf{A})$. This family is, in fact, $RM(\bar{\mathbf{A}})$.  
 \end{proof}

\begin{example}\label{NewExample}
The relational models generated by the matrices
$$
\mathbf{A}= \left(\begin{array}{rrrr}
1& 1& 1& 0\\
0& 0& 1& 1\\
\end{array}\right), \qquad \bar{\mathbf{A}}= \left(\begin{array}{rrrr}
1& 1& 1& 1\\
1& 1& 1& 0\\
0& 0& 1& 1\\
\end{array}\right)
$$
consist of positive probability distribution which can be written in the following forms:
$$\left\{\begin{array}{l}
p_1 = \alpha_1,\\
p_2 = \alpha_1,\\
p_3 = \alpha_1\alpha_2,\\
p_4 = \alpha_2, 
\end{array} \right.     \qquad 
\left\{\begin{array}{l}
p_1 = \beta_0\beta_1,\\
p_2 = \beta_0\beta_1,\\
p_3 = \beta_0\beta_1\beta_2,\\
p_4 = \beta_0\beta_2,
\end{array} \right. 
$$
where $\beta_0$ is the overall effect. The dual representations can be written in the log-linear form, using $\bd d_1 = (-1,0,1,-1)\tr , \bd d_2 = (-1,1,0,0)\tr \in Ker(\mathbf{A})$:
$$\left\{\begin{array}{l}
\bd d_1\tr \log  \bd p = 0,\\
\bd d_2\tr \log  \bd p = 0,
\end{array} \right.     \qquad 
\left\{\begin{array}{l}\bd d_2\tr \log  \bd p = 0. \\
\end{array} \right. 
$$
By Theorem \ref{OEnonhom}, after the overall effect is added,  the model specification does not include the non-homogeneous constraint anymore. In terms of the generalized odds ratios:
$$\left\{\begin{array}{l}
p_3 /( p_1p_4) = 1,\\
p_1 / p_2 = 1, 
\end{array} \right.     \qquad 
\left\{\begin{array}{l}p_1 / p_2 = 1. \\
\end{array} \right. 
$$
The second model may be defined using restrictions only on homogeneous odds ratios, and there is no need to place  an explicit restriction on the non-homogeneous odds ratio. 

\qed
\end{example}

The changes in the model structure after the overall effect is removed are examined next.  

A relational model with the overall effect can be reparameterized so that its model matrix has a row of $1$'s, and because of full row rank, this vector is not spanned by the other rows.  
The implications of the removal of the overall effect will be investigated using  a model matrix of this structure, say $\bar{\mathbf{A}}_1$. By the removal of the row  $\boldsymbol 1\tr$, one may obtain a different model matrix on the same sample space, but it may happen that there exists a cell $i_0$, whose only parameter is the overall effect, and after its removal, the $i_0$-th column  contains zeros only. In such cases, to have a proper model matrix, such columns, that is such  cells, need to be removed.   Write $\mathcal{I}_0$ for the set of all such cells $i_0$, and let $I_0 = |\mathcal{I}_0|$. Then, the reduced model matrix, $\mathbf{A}_1$, is obtained from $\bar{\mathbf{A}}_1$ after removing the row of $1$'s  and deleting the columns which, after this, contain only zeros. This is a model matrix on   $\mathcal{I} \setminus \mathcal{I}_0$.   Without loss of generality, the matrix $\bar{\mathbf{A}}_1$ can be written as:
$$\bar{\mathbf{A}}_1 = \left(\begin{array}{cc} \boldsymbol 1_{(I - I_0)}\tr &\boldsymbol 1_{I_0}\tr \\
 \mathbf{A}_1 &\mathbf{O}_{(J-1)\times I_0}  \end{array} \right).$$

If the sample spaces of $RM(\bar{\mathbf{A}}_1)$ and $RM({\mathbf{A}}_1)$ are the same that is, when $\mathcal{I}_0$  is empty, the reduced model is the subset of the original one, consisting of the distributions whose overall effect is zero, see Theorem \ref{newTheoremVarieties}. If the sample space is reduced, the relationship between the kernel basis matrices is described in the next result.

\begin{theorem}\label{thConjecture1}
The following holds:
\begin{enumerate}[(i)]
\item $dim(Ker(\mathbf{A}_1))=dim(Ker(\bar{\mathbf{A}}_1))-I_0+1$.

\item The kernel basis matrix $\mathbf{D}_1$ of $\mathbf{A}_1$ may be obtained from the kernel basis matrix  $\bar{\mathbf{D}}_1$ of  $\bar{\mathbf{A}}_1$ by deleting the the columns in $ \mathcal{I}_0$ and then leaving out the redundant rows.

\end{enumerate}
\end{theorem}

\begin{proof}
\begin{enumerate}[(i)]
\item Because $\bar{\mathbf{A}}_1$ is a $J \times I$ matrix of full row rank, $dim(Ker(\bar{\mathbf{A}}_1)) =  I-J$. The linear independence of its rows implies that the rows of $\mathbf{A}_1$ are also linearly independent. Therefore, because $\mathbf{A}_1$ is a $(J-1) \times (I-I_0)$ matrix, $dim(Ker(\mathbf{A}_1)) = I - I_0 - (J-1)$, which implies the result.  

\item Let  $\boldsymbol d_1, \boldsymbol d_2, \dots, {\boldsymbol d}_{I-J}$ be a kernel basis of $\bar{\mathbf{A}}_1$. Write 
$$\boldsymbol d_i = (\boldsymbol{u}_i\tr,\boldsymbol{v}_i\tr)\tr, \quad \mbox{ for } \quad i = 1,\dots, I-J.$$
Then,
$$\boldsymbol 0 = \bar{\mathbf{A}}_1\boldsymbol d_i  = \left(\begin{array}{cc} \boldsymbol 1_{(I - I_0)}\tr &\boldsymbol 1_{I_0}\tr \\
 \mathbf{A}_1 &\mathbf{O}  \end{array} \right) \left(\begin{array}{c}\boldsymbol{u}_i\\ \boldsymbol{v}_i\end{array}\right),  \quad \mbox{ for } \quad i = 1,\dots, I-J,$$
which implies that
\begin{eqnarray}\label{basisID}
\boldsymbol 1_{(I - I_0)}\tr\boldsymbol u_i + \boldsymbol 1_{I_0}\tr \boldsymbol v_i = 0, \quad \mathbf{A}_1 \boldsymbol u_i = \boldsymbol 0, \quad \mbox{ for } \quad i = 1,\dots, I-J.
\end{eqnarray}
Suppose $\mathbf{A}_1$ does not have the overall effect. Notice that each $\boldsymbol v_i$ has length $I_0$, and therefore, one can apply a non-singular linear transformation to the basis vectors $\boldsymbol d_1, \boldsymbol d_2, \dots, {\boldsymbol d}_{I-J}$ to reduce them to the form:
\begin{eqnarray*}
\boldsymbol d_1&=& (\boldsymbol u_1\tr, 1, 0,\dots, 0)\tr\\
\boldsymbol d_2 &=&(\boldsymbol u_2\tr, 0,  1,\dots, 0)\tr\\
\cdots&\\
\boldsymbol d_{I_0} &=&(\boldsymbol u_{I_0}\tr, 0,  0,\dots, 1)\tr\\
\boldsymbol d_{I_0+1} &=&(\boldsymbol u_{I_0+1}\tr, 0,  0,\dots, 0)\tr\\
\boldsymbol d_{I_0+2} &=&(\boldsymbol u_{I_0+2}\tr, 0,  0,\dots, 0)\tr\\
\cdots& \\
\boldsymbol d_{I-J} &=&(\boldsymbol u_{I-J}\tr, 0,  0,\dots, 0)\tr.
\end{eqnarray*}
The equations (\ref{basisID}) imply that 
\begin{equation*}
\boldsymbol 1_{(I - I_0)}\tr\boldsymbol u_i =-1, \,\, \mbox{ for } i = 1, \dots, I_0, \qquad \boldsymbol 1_{(I - I_0)}\tr\boldsymbol u_i = 0, \,\, \mbox{ for } i = I_0+1, \dots, I-J, 
\end{equation*}
$$ \mbox{ and }\quad \mathbf{A}_1 \boldsymbol u_i = \boldsymbol 0, \,\, \mbox{ for } i = 1, \dots, I-J.$$
The linear independence of $\boldsymbol d_{I_0+1}, \dots, \boldsymbol d_{I-J}$ in $\mathbb{R}^I$ entails the linear independence of 
$\boldsymbol u_{I_0+1}, \dots, \boldsymbol u_{I-J}$ in $\mathbb{R}^{I-I_0}$. Notice that  $\boldsymbol u_{1}, \dots, \boldsymbol u_{I_0}$ are jointly linearly independent from $\boldsymbol u_{I_0+1}, \dots, \boldsymbol u_{I-J}$, but not necessarily  linearly independent from each other. A kernel basis of $\mathbf{A}_1$ comprises $I-J-I_0+1$ linearly independent vectors in $Ker(\mathbf{A}_1)$, and, for example, $\boldsymbol u_{I_0},\boldsymbol u_{I_0+1}, \dots, \boldsymbol u_{I-J}$ form such a basis. Therefore, ${\mathbf{D}}_1$ can be derived from a kernel basis matrix of $\bar{\mathbf{A}}_1$ by removing the columns for $\mathcal{I}_0$ and leaving out the $I_0 - 1$ redundant rows.

Suppose $\mathbf{A}_1$ has the overall effect; without loss of generality, $\boldsymbol 1\tr$ is a row of  $\mathbf{A}_1$. In this case, (\ref{basisID}) implies that both $\boldsymbol 1_{(I - I_0)}\tr\boldsymbol u_i  = 0$ and  $\boldsymbol 1_{I_0}\tr \boldsymbol v_i = 0$, for $i = I_0+1, \dots, I-J$. Because $\boldsymbol u_i$'s and $\boldsymbol v_i$'s vary independently from each other, the linear independence of $\boldsymbol d_1, \boldsymbol d_2, \dots, {\boldsymbol d}_{I-J}$ will imply that $(\boldsymbol u_i, \boldsymbol 0)$, for $i = I_0+1, \dots, I-J$, are also  linearly independent in $\mathbb{R}^I$. Consequently, any $I-J-I_0+1$  vectors among $\boldsymbol u_i$'s are linearly independent in $\mathbb{R}^{I-I_0}$ and can form a kernel basis of $\mathbf{A}_1$.  Thus, as in the previous case, ${\mathbf{D}}_1$ can be derived from a kernel basis matrix of $\bar{\mathbf{A}}_1$ by removing the columns for $\mathcal{I}_0$ and leaving out the $I_0 - 1$ redundant rows.
\end{enumerate}
\end{proof}

The next two examples illustrate the theorem. 

\begin{example}\label{removeOEexample1}
Let $RM(\bar{\mathbf{A}}_1)$ be the relational model generated by 
$$
\bar{\mathbf{A}}_1= \left(\begin{array}{rrrrrr}
 1& 1& 1& 1 & 1&1\\
1& 0&  1& 0& 0 & 0\\
0& 1& 1& 1&0&0\\
\end{array}\right).
$$
Here, $\mathcal{I}_0 = \{5,6\}$. In terms of the generalized odds ratios the model can be written as:
$$\left\{\begin{array}{l}
p_3p_5 / (p_1 p_2)= 1,\\
p_3 p_6/( p_1p_2) = 1,\\
p_2/p_4 = 1.
 \end{array} \right. $$
 Remove the row $\boldsymbol 1\tr$ and the last two columns and consider the reduced matrix:
$$
{\mathbf{A}}_1= \left(\begin{array}{rrrr}
1& 0& 1&  0\\
0& 1& 1& 1\\
\end{array}\right).
$$
The model  $RM({\mathbf{A}}_1)$ does not have the overall effect and can be specified by two generalized odds ratios:
$$\left\{\begin{array}{l}
p_2 /( p_1p_4) = 1,\\
p_2 /p_4 = 1.
 \end{array} \right. $$
These odds ratios are defined on the smaller probability space, and may be obtained by  removing $p_5$ and $p_6$, and the redundant odds ratio, from the odds ratio specification of the original model.
\qed
\end{example}

\vspace{5mm}

\begin{example}\label{removeOEexample2}
Consider the relational model $RM(\bar{\mathbf{A}}_1)$ generated by 
$$
\bar{\mathbf{A}}_1= \left(\begin{array}{rrrr}
1 & 1& 1&  1\\
1& 1& 1&   0\\
1& 0& 1&  0\\
\end{array}\right).
$$
In terms of the generalized odds ratios, the model specification is $p_1/p_3 = 1$. Notice that $\bar{\mathbf{A}}_1$ is row equivalent to $$\bar{\mathbf{A}}_2= \left(\begin{array}{rrrr}
0 & 0& 0&  1\\
1& 1& 1&   0\\
1& 0& 1&  0\\
\end{array}\right).
$$
Because every $\boldsymbol d$ in  $Ker(\mathbf{A})$ is orthogonal to $(0, 0, 0, 1)$,  its last component has to be zero: $d_4 = 0$. Therefore, in any specification of $RM(\bar{\mathbf{A}}_1)$ in terms of the generalized odds ratios, $p_4$ will not be present. 
Set
$$
{\mathbf{A}}_1= \left(\begin{array}{rrr}
1& 1& 1\\
1& 0& 1\\
\end{array}\right).
$$
The model  $RM({\mathbf{A}}_1)$ has the overall effect and can be specified by exactly the same generalized odds ratio as the model $RM(\bar{\mathbf{A}}_1)$: $p_1/p_3 = 1$.

As a further illustration, take
$$
\bar{\mathbf{A}}_1= \left(\begin{array}{rrrrrr}
1 & 1& 1&  1&1&1\\
1& 1& 1&   0& 0& 0\\
1& 0& 1&  0 & 0& 0\\
\end{array}\right).
$$
In this case, 
$$
\bar{\mathbf{A}}_2= \left(\begin{array}{rrrrrr}
0 & 0& 0&  1&1&1\\
1& 1& 1&   0& 0& 0\\
1& 0& 1&  0 & 0& 0\\
\end{array}\right),
$$
and $\mathbf{A}_1$ is the same as above. In this case, $RM(\bar{\mathbf{A}}_1)$ is specified by
$p_1/p_3 = 1, p_4/p_5 = 1, p_4/p_6 = 1$, and $RM(\bar{\mathbf{A}}_1)$ is described as previously: $p_1/p_3 = 1$.
\qed
\end{example}


The polynomial variety $\mathcal{X}_{\bar{\mathbf{A}}_1}$ defining the model $RM(\bar{\mathbf{A}}_1)$ is homogeneous. If the removal of the cells comprising $\mathcal{I}_0$ leads to a model without the overall effect, the variety  $\mathcal{X}_{\bar{\mathbf{A}}_1}$ is dehomogenized, yielding the affine variety $\mathcal{X}_{{\mathbf{A}}_1}$ \citep[cf.][]{Cox}. 

The converse to this procedure, homogenization of an affine variety, is also studied in algebraic geometry, and is performed by introducing a new variable in such a way that all equations defining the variety become homogeneous. The essence of this procedure is that all probabilities are multiplied by this new variable. This leaves the homogeneous odds ratios unchanged, as the new variable cancels out. The value of a non homogeneous odds ratio becomes, instead of $1$, the reciprocal of the new variable. For example, the odds ratio in Example \ref{removeOEexample1}
$p_2/(p_1p_4)=1$ becomes  $vp_2/(vp_1vp_4)=1/v$, where $v$ is the new variable. If now $v$ could be seen as the probability of an additional cell, say $p_0$, then this would be a homogeneous odds ratio, $p_0p_2/(p_1p_4)=1$.

Although a straightforward procedure in algebraic geometry, it does not necessarily have a clear interpretation in statistical inference. Introducing a new variable and a new cell for the purpose of homogenization can be made meaningful in some situations, if the sample space may be extended by one cell, and the new variable is the parameter (probability) of this cell. Homogenization requires this new variable to appear in every cell, too, so the parameter may be seen as the overall effect. The new cell has only the overall effect, thus no feature is present in this cell.

The augmentation of the sample space by an additional cell does make sense, if that cell exists in the population but was not observed because of the design of the data collection procedure, as in Example \ref{crabs}. The additional cell has the overall effect only, thus is a ``no feature present'' cell.

\begin{example}\label{crabs} 
In the study of swimming crabs by \cite*{Kawamura1995}, three types of baits were used in traps to catch crabs: fish alone, sugarcane alone, fish-sugarcane combination. The sample space consists of three cells, $\mathcal{I} = \{(0, 1), (1,0), (1,1)\}$, and the cell $(0,0)$ is absent by design, because there were no traps without any bait. Under the AS independence, the cell parameter associated with both bait types present is the product of the parameters associated with the other two cells. This is a relational model without the overall effect, generated by the matrix \citep[cf.][]{KRipf1}:
$$\mathbf{A} = \left(
\begin{array}{ccc}
1 & 0 & 1 \\
0 & 1 & 1 \\
\end{array}
\right).$$
In fact, the overall effect cannot be included in this situation, because it would saturate the model. The affine variety associated with this model can be homogenized by including a new variable. The new variable may only be interpreted as the parameter associated with no bait present, and calls for an additional cell in the sample space (to avoid model the saturation of the model) which  may only be interpreted as setting up a trap without any bait, which would also be a plausible research design.  The resulting model is generated by $\mathbf{A}_0$:
$$\mathbf{A}_0 = \left(
\begin{array}{cccc}
1& 1& 1& 1\\
0& 1 & 0 & 1 \\
0& 0 & 1 & 1 \\
\end{array}
\right),$$
and indeed, is the model of traditional independence on the complete $2 \times 2$ contingency table. \qed
\end{example}


For situations like in Example \ref{crabs},  the AS independence is a natural model, but it also applies to cases, when the ``no feature present'' situation is logically impossible (like market basket analysis, or records of traffic violations, see \cite{KRD11, KRipf1}, and also the biological example in Section \ref{hemato}), and in such cases, the cell augmentation procedure is not meaningful. There are, however situations, when the existence of the ``no feature present'' cell is logically not impossible, but the actual existence in the population is dubious.

For a more general discussion of the homogenization of AS independence,
let $\boldsymbol d_1, \dots, \boldsymbol d_K$ be a kernel basis of $\mathbf{A}$, satisfying (\ref{KernelRows}) with $\boldsymbol d_1\tr \boldsymbol 1 = -1$. The polynomial ideal $\mathscr{I}_{\mathbf{A}}$ associated with the matrix $\mathbf{A}$ is generated by one non-homogeneous polynomial $\boldsymbol p^{\boldsymbol d_1^+} - \boldsymbol p^{\boldsymbol d_1^-}$, and $K-1$ homogeneous polynomials, $\boldsymbol p^{\boldsymbol d_k^+} - \boldsymbol p^{\boldsymbol d_k^-}$. Notice that, because $(\boldsymbol d_1^+)\tr \boldsymbol 1 - (\boldsymbol d_1^-)\tr \boldsymbol 1$ $=  (\boldsymbol d_1^+ - \boldsymbol d_1^-)\tr \boldsymbol 1$ = $\boldsymbol d_1\tr \boldsymbol 1$, the difference in the degrees of the monomials $p^{\boldsymbol d_1^+}$ and $p^{\boldsymbol d_1^-}$ is $-1$. Therefore, the polynomial  $\boldsymbol p^{\boldsymbol d_1^+} - \boldsymbol p^{\boldsymbol d_1^-}$ can be homogenized by multiplying the first monomial by one additional variable, say $p_0$:
$$p_0\boldsymbol p^{\boldsymbol d_1^+} - \boldsymbol p^{\boldsymbol d_1^-}.$$ 
The polynomial ideal generated by
$$p_0\boldsymbol p^{\boldsymbol d_1^+} - \boldsymbol p^{\boldsymbol d_1^-}, \boldsymbol p^{\boldsymbol d_2^+} - \boldsymbol p^{\boldsymbol d_2^-}, \dots, \boldsymbol p^{\boldsymbol d_K^+} - \boldsymbol p^{\boldsymbol d_K^-}$$
and the corresponding variety are homogeneous, and can be described by the matrix of size $(J+1)\times (I+1)$ of the following structure:
$${\mathbf{A}}_0 = \left(\begin{array}{cc} 1 & \boldsymbol 1\tr_I \\
\boldsymbol 0_J & \mathbf{A} \end{array} \right).$$
Here, $\boldsymbol 1\tr_I$ is the row of $1$'s of length $I$, and $\boldsymbol 0_J$ is the column of zeros of length $J$. 

In fact, the homogeneous variety $\mathcal{X}_{{\mathbf{A}}_0}$ is the projective closure of the affine variety $\mathcal{X}_{\bar{\mathbf{A}}_0}$ \citep{Cox}. The latter can be obtained from the former by dehomogenization via setting $p_0 = 1$.

The homogenization of the model of AS independence for three features is discussed next.

\begin{example}\label{Example0}
Consider the model of AS independence for three attributes, $A$, $B$, and $C$, described in \cite{KRD11}. 
\begin{equation}\label{indepAtr}
{p_{110}}={p_{100}p_{010}}, \,\,{ p_{101}}={p_{100}p_{001}}, \,\, {p_{011}}={p_{010}p_{001}}, \,\,{p_{111}}={p_{100}p_{010}p_{001}}.
\end{equation}
Here $p_{ijk} = \mathbb{P}(A = i, B=j, C=k)$ for $i, j, k \in \{0,1\}$, where the combination $(0,0,0)$ does not exist, and $\sum_{ijk} p_{ijk} = 1$. The equations (\ref{indepAtr}) specify  the relational model generated by
$$
\mathbf{A}= \left(\begin{array}{rrrrrrr}
1& 0& 0& 1& 1& 0& 1\\
0& 1& 0& 1& 0& 1& 1\\
0& 0& 1& 0& 1& 1& 1\\
\end{array}\right).
$$ 
Consider the following kernel basis matrix which is of the form (\ref{KernelRows}):
$$\mathbf{D} = \left(\begin{array} {rrrrrrr}
-1&-1&0&1&0&0&0\\
-1&0&1&1&0&-1&0 \\
0&-1&1&1&-1&0&0\\ 
0&0&-1&0&1&1&-1 \\
\end{array}\right).$$
The corresponding polynomial ideal is:
$$\mathscr{I}_{\mathbf{A}} =\langle \,\, p_{110} - p_{100} p_{010}, \,\,p_{110}p_{001} - p_{011}p_{100}, \,\, p_{110}p_{001} - p_{101}p_{010}, \,\, p_{111}p_{001} - p_{101}p_{011} \,\, \rangle.$$
The generating set of $\mathscr{I}_{\mathbf{A}}$ includes at least one non-homogeneous polynomial, due to $\boldsymbol d_1$, and  can be homogenized by introducing a new variable, say $p_{000}$.   The resulting ideal,
$$\mathscr{I}_{\mathbf{A}_0} = \langle \,\, p_{000}p_{110} - p_{100} p_{010}, \,\,p_{110}p_{001} - p_{011}p_{100}, \,\, p_{110}p_{001} - p_{101}p_{010}, \,\, p_{111}p_{001} - p_{101}p_{011} \,\, \rangle,$$
is homogeneous, and its zero set
\begin{equation}\label{variety0}
\mathcal{X}_{\mathbf A_0} = \left\{\boldsymbol p \in \mathbb{R}^{|\mathcal{I}+1|}_{\geq 0}: \,\, \boldsymbol p^{\boldsymbol d^+} = \boldsymbol p^{\boldsymbol d^-}, \,\,  \forall  \boldsymbol d \in Ker(\mathbf{A}_0) \right \},
\end{equation}
where 
$$
\mathbf{A}_0= \left(\begin{array}{rrrrrrrr}
1& 1& 1& 1& 1& 1& 1& 1\\
0& 1& 0& 0& 1& 1& 0& 1\\
0& 0& 1& 0& 1& 0& 1& 1\\
0& 0& 0& 1& 0& 1& 1& 1\\
\end{array}\right),
$$
is thus a homogeneous variety. The relational model $RM(\mathbf{A}_0)$ is defined on a larger sample space, namely $\mathcal{I}\cup(0,0,0)$.  The model has the overall effect and is  the following set of distributions: 
\begin{equation}\label{ExtMprob0}
\boldsymbol p \in \mathcal{X}_{\mathbf A_0} \cap int(\Delta_{I}).
\end{equation} 
The rows of $\mathbf{A}_0$ are the indicators of the cylinder sets of the total (the row of $1$'s), and of the $A$, $B$, and $C$ marginals. Therefore, the relational model $RM(\mathbf{A}_0)$ is the traditional model of mutual independence. 
\qed
\end{example}

The next theorem states in general what was seen in the example. Let $X_1, \dots, X_T$ be the random variables taking values in $\{0,1\}$. Write $\mathcal{I}^0$ for the Cartesian product of their ranges,  and let $\mathcal{I} = \mathcal{I}^0 \setminus (0,\dots,0) $.

\begin{theorem}
Let $RM(\mathbf{A})$ be the model of AS independence of $X_1, \dots, X_T$ on the sample space $\mathcal{I}$. The interior of the projective closure of this model is the log-linear model  of mutual independence of  $X_1, \dots, X_T$ on the sample space $\mathcal{I}^0$. 
\end{theorem}

\begin{proof}
Let $\mathbf{A}$ be the model matrix for the AS independence:
$$
\mathbf{A}= \left(\begin{array}{rrrrrrrrr}
1& 0& 0&\dots & 0& 1& 1& \dots &1\\
0& 1&  0&\dots& 0& 1& 0& \dots&1\\
0& 0&  1&\dots& 0& 0& 1& \dots &1\\
 \vdots& \vdots & \vdots  &\ddots & \vdots & \vdots& \vdots & \ddots & \vdots \\
0& 0& 0&\dots&1& 0& 0& \dots &1\\
\end{array}\right).
$$
The number of columns of $\mathbf{A}$ is equal to the number of cells in the sample space $\mathcal{I}$, $I = 2^T - 1$. 
The model $RM(\mathbf{A})$ is  the intersection of the polynomial variety $\mathcal{X}_{\mathbf{A}}$ and the interior of the simplex $\Delta_{I - 1}$. The variety $\mathcal{X}_{\mathbf{A}}$ is non-homogeneous, because among its generators there is at least one non-homogeneous polynomial. In order to obtain the projective closure of $\mathcal{X}_{\mathbf{A}}$ \citep[cf.][]{Cox}, include the ``no feature present'' cell, indexed by $0$, to the sample space, choose a Gr{\"o}bner basis of the ideal $\mathscr{I}_{\mathbf{A}}$, and homogenize all non-homogeneous polynomials in this basis using the cell probability $p_{0}$. Because the projective closure of $\mathcal{X}_{\mathbf{A}}$ is the minimal homogeneous variety in the projective space whose dehomogenization is $\mathcal{X}_{\mathbf{A}}$ \citep{Cox}, Theorem \ref{thConjecture1}(ii) implies that this closure  can be described using the matrix  
$$\mathbf{A}_0= \left(\begin{array}{rrrrrrrrrrr}
1&1 & 1& 1& \dots & 1& 1& 1& \dots& 1\\
0&1& 0& 0&\dots & 0& 1& 1& \dots &1\\
0&0& 1&  0&\dots& 0& 1& 0& \dots&1\\
0& 0& 0&  1&\dots& 0& 0& 1& \dots &1\\
\vdots& \vdots& \vdots & \vdots  &\ddots & \vdots & \vdots& \vdots & \ddots & \vdots \\
0& 0& 0& 0&\dots&1& 0& 0& \dots &1\\
\end{array}\right).$$ 
Each distribution in $RM(\mathbf{A})$ has the multiplicative structure prescribed by $\mathbf{A}$ \citep{KRextended}, and during the homogenization, is mapped in a positive distribution in $\mathcal{X}_{\mathbf{A}_0}$. Because all strictly positive distributions in $\mathcal{X}_{\mathbf{A}_0}$ have the multiplicative structure prescribed by $\mathbf{A}_0$, they comprise the relational model  $RM(\mathbf{A}_0)$. This matrix describes the model of mutual independence between $X_1, \dots, X_T$ in the effect coding, and the proof is complete. 
\end{proof}

The homogenization (in the language of algebraic geometry) or regularization (in the language of the exponential families) leads to a simpler structure, which allows a simpler calculation of the MLE. However, the additional cell was not observed in these cases, and assuming its frequency is zero is ungrounded and may lead to wrong inference. 

The framework developed here may also be used to define context specific independence, so that in one context conditional independence holds, in another one, AS independence does. To illustrate, let $X_1$, $X_2$, $X_3$  be random variables taking values in $\{0,1\}$. Assume that the $(0,0,0)$ outcome is impossible, so the sample space can be expressed as:
$$
\begin{tabular}{ccccc}
   & \multicolumn{2}{c}{$X_3=0$}&\multicolumn{2}{c}{$X_3=1$}\\
\cmidrule(lr){2-3}  \cmidrule(lr){4-5} 
   & $X_2=0$ & $X_2=1$ & $X_2=0$ & $X_2=1$ \\ [3pt]
\hline \\ 
$X_1=0$ &- & $p_{010}$ & $p_{001}$ & $p_{011}$\\
$X_1=1$ & $p_{100}$  & $p_{110}$ & $p_{101}$ & $p_{111}$\\
\end{tabular} 
$$
Let $\boldsymbol p = (p_{001}, p_{010}, p_{011}, p_{100}, p_{101}, p_{110}, p_{111})$, and consider the relational model without the overall effect generated by 
\begin{equation}\label{modMxCIAS}
\mathbf{A}_0 = \left(\begin{array}{ccccccc}
0&0&0&1&1&1&1\\
0&1&1&0&0&1&1\\
1&0&1&0&1&0&1\\
0&0&0&0&1&0&1\\
0&0&1&0&0&0&1\\
\end{array} \right),
\end{equation}
The kernel basis matrix is equal to:
\begin{equation}\label{kernelCIconventional}
\mathbf{D}_0 = \left(\begin{array}{rrrrrrr}
0& -1& 0& -1& 0& 1&0 \\
1&0&-1&0&-1&0&1\\
\end{array} \right),
\end{equation}
and thus, the model can be specified in terms of the following two generalized odds ratios:
$$\mathcal{COR}(X_1X_2 \mid X_3 = 0) = \frac{p_{110}}{p_{010}p_{100}} = 1, \qquad  \mathcal{COR}(X_1X_2 \mid X_3 = 1) = \frac{p_{001}p_{111}}{p_{011}p_{101}} = 1.$$
The second constraint expresses the (conventional) context-specific independence of $X_1$ and $X_2$ given $X_3 = 1$. The first odds ratio is non-homogeneous, and the corresponding constraint may be seen as the context-specific AS-independence of of $X_1$ and $X_2$ given $X_3 = 0$.

\section{ML estimation with and without the overall effect}\label{MLEsection}

The properties of the ML estimates under relational models, discussed in detail in  \cite{KRD11} and \cite{KRextended},  are summarized here in the language of the linear and multiplicative families defined by the model matrix and its kernel basis matrix. The conditions of existence of the MLE are reviewed first.

Let $\boldsymbol a_1, \dots, \boldsymbol a_{|\mathcal{I}|}$ denote the columns of $\mathbf{A}$, and let $C_{\mathbf{A}} = \{ \boldsymbol t \in \mathbb{R}^J_{\geq 0}: \,\, \exists \boldsymbol p \in \mathbb{R}^{|\mathcal{I}|}_{\geq 0}  \quad \boldsymbol t = \mathbf{A}\boldsymbol p\}$  be the polyhedral cone whose relative interior comprises such $\boldsymbol t \in \mathbb{R}^J_{> 0}$, for which there exists a $\boldsymbol p > \boldsymbol 0$ that satisfies $\boldsymbol t = \mathbf{A}\boldsymbol p$. A set of indices $F = \{i_1, i_2, \dots,i_f\}$ is called  facial if the columns $\boldsymbol a_{i_1}, \boldsymbol a_{i_2}, \dots, \boldsymbol a_{i_f}$ are affinely independent and span a proper face of $C_{\mathbf{A}}$  \citep*[cf.][]{GrunbaumConvex, GeigerMeekSturm2006,FienbergRinaldo2012}.  It can be shown that a set $F$ is facial if and only if there exists a $\boldsymbol c \in \mathbb{R}^J$, such that $\boldsymbol c'\boldsymbol a_i = 0$ for every $i \in F$ and $\boldsymbol c'\boldsymbol a_i  > 0$ for every $i \notin F$.

Let $\boldsymbol q \in \mathcal{P}$ and let  $\cg K$ be the set of $\kappa > 0$, such that, for a fixed $\kappa$, the linear family
\begin{equation}\label{LinFam}
\cg F(\mathbf A, \bd q,\kappa) = \{\bd r \in \cg P: \,\, \mathbf{A} \bd r = \kappa\mathbf{A}\bd q\}
\end{equation}
is not empty, and let  $\cg F(\mathbf A, \bd q) = \bigcup_{\cg K} \cg F(\mathbf A, \bd q,\kappa)$. For each $\kappa > 0$, the linear family $\cg{F}(\mathbf{A}, \bd q, \kappa)$ is a polyhedron in the cone $C_{\mathbf{A}}$.

\begin{theorem}\label{MLEextendTHnew}\citep{KRextended}
Let  $RM(\mathbf{A})$ be a relational model, with or without the overall effect, and let $\boldsymbol q$ be the observed distribution.
\begin{enumerate}
\item  The MLE $\hat{p}_{\boldsymbol q}$ given $\boldsymbol q$ exists if only:
\begin{enumerate}[(i)]
\item $supp(\boldsymbol q) = \mathcal{I}$, or
\item $supp(\boldsymbol q) \subsetneq \mathcal{I}$ and, 
 for all facial sets $F$ of $\mathbf{A}$, $supp(\boldsymbol q) \not\subseteq F$. 
\end{enumerate}
In either case, $\hat{\bd p}_q = \cg F(\mathbf A, \bd q)  \cap  int(\mathcal{X}_{\mathbf{A}})$,  and there exists a unique constant $\gamma_q > 0$, also depending on $\mathbf A$, such that:
$$ \mathbf{A}\hat{\bd p}_q = \gamma_q \mathbf{A} \bd q,  \quad \bd 1\tr \hat{\bd{p}}_q = 1.$$
\item The MLE under the extended model $\,\xbar{RM}(\mathbf{A})$, defined in (\ref{ExtMprob}), always exists and is the unique point of $\mathcal{X}_{\mathbf{A}}$ which satisfies:
\begin{align}
&\mathbf{A}\boldsymbol{p} = \gamma_q \mathbf{A } \boldsymbol q, \mbox{ for some } \gamma_q > 0; \label{A} \\
&\boldsymbol 1'\boldsymbol p = 1. \nonumber
\end{align}
\end{enumerate}
\end{theorem}

The statements follow from Theorem 4.1 in \cite{KRextended} and Corollary 4.2 in \cite{KRD11}, and the proof is thus omitted.  The constant $\gamma_q$, called  the adjustment factor, is the ratio between the subset sums of the MLE, $\mathbf{A}\hat{\bd p}_q$, and the subset sums of the observed distribution, $\mathbf{A}\bd q$. If the overall effect is present in the model, $\gamma_q = 1$ for all $\boldsymbol q$.

Let $\mathbf{A}$ be a model matrix whose row space does not contain $\boldsymbol 1\tr$, and let $\bar{\mathbf{A}}$ be the matrix obtained by augmenting $\mathbf{A}$ with the row $\boldsymbol 1\tr$. It will be shown in the proof of the next theorem that every facial set of $\mathbf{A}$ is facial for $\bar{\mathbf{A}}$. If the observed $\boldsymbol q$ is positive, the MLEs  $\hat{\bd p}_{q}$ and  $\bar{\bd p}_q$ under the models $RM({\mathbf A})$ and $RM(\bar{\mathbf A})$, respectively, exist. However, as implied by the relationship between  the facial sets of ${\mathbf A}$ and $\bar{\mathbf{A}}$, if $\boldsymbol q$ has some zeros,  the MLE may exist under $RM({\mathbf A})$, but not under $RM(\bar{\mathbf A})$, or neither of the MLEs exist.

\begin{theorem}
Let $\mathbf{A}$ be a model matrix whose row space does not contain $\boldsymbol 1\tr$, and let $\bar{\mathbf{A}}$ be the matrix obtained by augmenting $\mathbf{A}$ with the row $\boldsymbol 1\tr$. Let $\boldsymbol q$ be the observed distribution. If, given $\boldsymbol q$, the MLE under $RM(\bar{\mathbf A})$ exists, so does the MLE under $RM({\mathbf A})$.
\end{theorem}
\begin{proof}
If $\boldsymbol q > \boldsymbol 0$, both MLEs exists.

Assume that $\boldsymbol q$ has some zeros, that is, $supp(\boldsymbol q) \subsetneq \mathcal{I}$, and that the MLE under $RM(\bar{\mathbf{A}})$ exists. It will be shown next that for any facial set $F$ of $\mathbf{A}$, $supp(\boldsymbol q) \not\subseteq F$.

The proof is by contradiction. Let $F_0$ be a facial set of $\mathbf{A}$, such that $supp(\boldsymbol q) \subset F_0$.  Therefore, there exists a $\boldsymbol c \in \mathbb{R}^J$, such that $\boldsymbol c'\boldsymbol a_i = 0$ for every $i \in F$ and $\boldsymbol c'\boldsymbol a_i  > 0$ for every $i \notin F$. 

Denote by $\bar{\boldsymbol a}_1, \dots, \bar{\boldsymbol a}_I$ the columns of $\bar{\mathbf{A}}$. By construction,  $\bar{\boldsymbol a}_i = (1, \boldsymbol a_i\tr)\tr$, $\,i = 1, \dots, I$. Let $\bar{\boldsymbol c} = (0,\boldsymbol c\tr)\tr$. Then,  
$$\bar{\boldsymbol c}\tr\bar{\boldsymbol a}_i = 0\cdot 1+\boldsymbol c'\boldsymbol a_i = \left\{ \begin{array}{l} \boldsymbol c'\boldsymbol a_i = 0, \quad \mbox{ for } i \in F_0, \\
\bar{\boldsymbol c}\tr\bar{\boldsymbol a}_i  > 0, \quad \mbox{ for  } i \notin F_0, \end{array} \right.
$$
and thus, $F_0$ is a facial set of $\bar{\mathbf{A}}$. Because $supp(\boldsymbol q) \subset F_0$, the MLE under $RM(\bar{\mathbf{A}})$, given $\boldsymbol q$, does not exist, which contradicts the initial assumption. This completes the proof.
\end{proof}

\vspace{5mm}

\textbf{Example \ref{Example0}} (revisited):
Let $\boldsymbol q_1 = (0,0,0,0,0,0,1)'$ be the observed distribution. Because  $supp(\boldsymbol q_1) = \{7\}$ is not a subset of any facial sets of $\mathbf{A}$, the MLE exists:
$$\hat{\boldsymbol p}_{q_1} = \left(\sqrt[3]{2} - 1, \sqrt[3]{2} - 1, \sqrt[3]{2} - 1, (\sqrt[3]{2} - 1)^2, (\sqrt[3]{2} - 1)^2, (\sqrt[3]{2} - 1)^2, (\sqrt[3]{2} - 1)^3\right)',$$ 
with $\hat{\gamma}_q = 2 - \sqrt[3]{4}$.  

On the other hand,  the set of indices $F = \{1,4,5,7\}$ is facial for $\bar{\mathbf{A}}$, and $supp(\boldsymbol q_1) \subsetneq F$. In this case, the MLE exists only in the extended model $\,\xbar{RM}(\bar{\mathbf{A}})$, and is equal to $\boldsymbol q_1$ itself.  

Let $\boldsymbol q_2 = (1,0,0,0,0,0,0)'$. Because $supp(\boldsymbol q_2) = \{1\}$ is a subset of a facial set of $\mathbf{A}$ and of a facial set of $\bar{\mathbf{A}}$, the MLEs exist only in the corresponding extended models. 
 \qed

 Further properties of the adjustment factor, including its geometrical meaning, are described next, relying on the following result: 
\begin{theorem}\label{ThFiber}
Let $\mathbf{A}$ be a model matrix whose row space does not contain $\boldsymbol 1\tr$, and let $\bar{\mathbf{A}}$ be the matrix obtained by augmenting $\mathbf{A}$ with the row $\boldsymbol 1\tr$.  For any $\bd r_1,\:\bd r_2 \in \mathcal{P}$, $\bd r_1 \neq \bd r_2$, the following holds:
\begin{enumerate}
\item \begin{enumerate}[(i)]
\item  The MLEs under $RM(\mathbf A)$, given they exist, are equal if and only if the subset sums entailed by $\mathbf{A}$ are proportional:
$$
\hat{\bd p}_{r_1} = \hat{\bd p}_{r_2}  \quad \Leftrightarrow  \quad \mathbf{A} \bd r_1 = \kappa \mathbf{A} \bd r_2 \quad \mbox{for some } \kappa\in \cg K
$$
and the adjustment factors in the MLE satisfy: $\kappa\hat\gamma_{r_1}$ = $\hat\gamma_{r_2}$.
\end{enumerate}
\item  The MLEs under $RM(\bar{\mathbf A})$, given they exist, are equal if and only if the subset sums entailed by $\mathbf{A}$ coincide:
$$
\bar{\bd p}_{r_1} = \bar{\bd p}_{r_2}  \quad \Leftrightarrow  \quad \mathbf{A} \bd r_1 = \mathbf{A} \bd r_2.
$$
\end{enumerate}
\end{theorem}

\vspace{6mm}

The statements are a reformulation of  Theorem 4.4 in \cite{KRD11}, and no proofs are provided here. The relationship between the adjustment factors is obvious.

The theorem implies that $\cg{F}(\mathbf A, \bd q)$ is an equivalence class in $\mathcal{P}$, in the sense that,  for any $\bd r \in \cg{F}(\mathbf A, \bd q)$, the MLE under $RM(\mathbf{A})$ satisfies $\hat{\bd p}_{r}$ = $\hat{\bd p}_{q}$. Each sub-family $\cg F(\mathbf A, \bd q,\kappa)$ is characterized by its unique adjustment factor under $RM(\mathbf{A})$. That is, for every $ \boldsymbol r_1, \boldsymbol r_2 \in \mathcal{F}(\mathbf{A}, \boldsymbol q, \kappa)$, $\,\,\boldsymbol r_1 \neq \boldsymbol r_2$,
$$\hat{\boldsymbol p}_{r_1} = \hat{\boldsymbol p}_{r_2} = \hat{\boldsymbol p}_{q}, \qquad \hat\gamma_{r_1} = \hat\gamma_{r_2} =\hat\gamma_{q}/\kappa.$$ 
In addition, $\,\,\bar{\boldsymbol p}_{r_1} =  \bar{\boldsymbol p}_{r_2}$ for any $ \boldsymbol r_1, \boldsymbol r_2 \in \mathcal{F}(\mathbf{A}, \boldsymbol q, \kappa)$, and therefore, for a fixed $\kappa$, $\cg F(\mathbf A, \bd q,\kappa)$ is an equivalence class under $RM(\bar{\mathbf{A}})$.

From a geometrical point of view, $\cg{F}(\mathbf{A}, \bd q)$ is a polyhedron which decomposes into polyhedra $\cg F(\mathbf{A}, \bd q,\kappa)$, with $\kappa>0$; clearly, $\bd q\in \cg F(\mathbf{A},\bd q,1)$. The MLE under $RM(\bar{\mathbf A})$ given $\bd r\in \cg F(\mathbf{A}, \bd q,\kappa)$ is the unique point common to the polyhedron $\cg F(\mathbf{A}, \bd q,\kappa)$ and the variety $\mathcal{X}_{\bar{\mathbf{A}}}$. Among the feasible values of $\kappa$ there exists a unique one, say $\hat\kappa$, such that the MLE $\bar{\bd p}_{r}$, $\forall\bd r\in \cg F(\mathbf{A}, \bd q,\hat \kappa)$, coincides with the MLE of $\bd q$ under $RM(\mathbf{A})$, $\hat{\bd p}_{q}$. This happens when $\hat\gamma_{r}=1$ so that, from (ii) in Theorem \ref{ThFiber}, $\hat\kappa$ = $\hat\gamma_{q}$. This latter point, $\hat{\bd p}_{q}$, is the intersection between $\cg F(\mathbf{A}, \bd q)$ and the non-homogeneous variety $\mathcal{X}_{{\mathbf{A}}}$. This specific value of the adjustment factor $\gamma_q = \hat\kappa$, is the adjustment factor of the MLE under $RM(\mathbf{A})$ given $\bd q$. An illustration is given next.

\vspace{5mm}

Relational models for probabilities without the overall effect are curved  exponential families, and the computation of the MLE under such models is not straightforward. An extension of the iterative proportional fitting procedure, G-IPF, that can be used for both models with and models without the overall effect was proposed in \cite{KRipf1} and is implemented in \cite{gIPFpackage}. Alternatively, the MLEs can be computed, for instance, using  the Newton-Raphson algorithm or the algorithm of \cite{EvansForcina11}. One of the algorithms described in \cite{Forcina2017} gave an idea of a possible modification of G-IPF. A brief description of the original and modified versions of G-IPF is given below:

\vspace{8mm}
\begin{center}
\begin{tabular}{lcr}
G-IPF & & G-IPFm \\
\hline 
 & &\\
Fix $\gamma> 0$  & &  Fix $\gamma > 0$ \\
 & &\\
\multicolumn{3}{c}{Run IPF($\gamma$) to obtain $\boldsymbol p_{\gamma}$, where}  \\
 & &\\
$\mathbf{A} \boldsymbol p_{\gamma} =  \gamma \mathbf{A}\boldsymbol q$ & & $\bar{\mathbf{A}} \boldsymbol p_{\gamma} = \left(\begin{array}{c} 1 \\ \gamma \mathbf{A}\boldsymbol q\end{array}\right)$ \\
 & & \\
$\mathbf{D} \log \boldsymbol p_{\gamma} = \boldsymbol 0$ & & $\bar{\mathbf{D}} \log \boldsymbol p_{\gamma} = \boldsymbol 0$\\
& &\\
\multicolumn{3}{c}{Adjust $\gamma$,  to approach the solution of} \\
 & &\\
$\boldsymbol 1\tr\boldsymbol p_{\gamma}  = 1$ & & $ \boldsymbol d_1\tr \log \boldsymbol p_{\gamma} = 0$ \\
& &\\
\multicolumn{3}{c}{{Iterate with the new $\gamma$ }} \\
 & &\\
\end{tabular}
\end{center}

\begin{theorem}\label{newGIPFconvTh}
If  $\boldsymbol q > \boldsymbol 0$, the G-IPFm algorithm converges, and its limit is equal to $\hat{\boldsymbol{p}}_q$, the ML estimate of $\boldsymbol p$ under $RM(\mathbf{A})$.
\end{theorem}

\begin{proof}
The convergence of one iteration of G-IPFm, when $\gamma$ is fixed, can be proved similarly to Theorem 3.2 in \cite{KRipf1}. The limit is positive, $\tilde{\boldsymbol p}_{\gamma} > \boldsymbol 0$,  and thus, by Lemma 1 in \cite{Forcina2017}$f(\gamma) = \boldsymbol d_1\tr \log{\tilde{\bd p}}_{\gamma}$ is a strictly increasing and differentiable function of $\gamma$. So, one can update $\gamma$, until for some ${\gamma}_q$  the G-IPFm limit satisfies: $f(\gamma_q) =\boldsymbol d_1 \log \tilde{\boldsymbol p}_{{\gamma}_q} = 0$. 
Because, in this case, 
$$ \bd A \tilde{\bd{p}}_{\gamma_q} =  {\gamma}_q \bd A\bd q, \quad \bd D \log \tilde{\bd{p}}_{\gamma_q} = \bd 0, \quad \bd 1\tr \tilde{\bd{p}}_{\gamma_q} =1,$$
the uniqueness of the MLE implies that $\tilde{\bd{p}}_{\gamma_q} = \hat{\bd p}_q$ and ${\gamma}_{q} = \hat{\gamma}_q$.
\end{proof}

The original G-IPF can be used whether or not $\boldsymbol q$ has some zeros, and it computes a sequence whose elements are the unique intersections of the variety $\mathcal{X}_{\mathbf{A}}$ and each of the polyhedra defined by $\mathbf{A}\tilde{\bd \tau}  = \gamma \mathbf{A} \bd q$ for different $\gamma$. This sequence converges, and its limit belongs to the hyperplane $\boldsymbol 1\tr \boldsymbol \tau = 1$ \citep{KRextended}.  G-IPFm produces a sequence whose elements are the unique intersections of  the  interior of the homogeneous variety $\mathcal{X}_{\bar{\mathbf{A}}}$ and each of the polyhedra $\cg F(\mathbf{A}, \bd q,\gamma)$. The limit of this sequence belongs to the interior of the non-homogeneous variety $\mathcal{X}_{\mathbf{A}}$. To ensure the existence, differentiability, and monotonicity of $f(\gamma)$, described above, the G-IPFm algorithm should   be applied only when $\boldsymbol q > \boldsymbol 0$. If $\boldsymbol q$ has some zero components, the positive MLE $\hat{\bd p}_q$ may  still exist, see Theorem \ref{MLEextendTHnew}(ii). However,  for some $\boldsymbol q$, because, in general, the matrices $\mathbf{A}$ and $\bar{\mathbf{A}}$ have different facial sets, no strictly positive $\boldsymbol p_{\gamma}$ would satisfy $\bar{\mathbf{A}} \boldsymbol p_{\gamma} = \left(\begin{array}{c} 1 \\ \gamma \mathbf{A}\boldsymbol q\end{array}\right)$.  

\vspace{5mm}

Some limitations and advantages of using the generalized IPF were addressed in \cite{KRipf1}, Section 2. In particular, while the assumption of the model matrix to be of full row rank can be relaxed for G-IPF, it is one of the major assumptions for the Newton-Raphson and the Fisher scoring algorithms. The algorithms proposed in \cite{Forcina2017} also require the model matrix to be of full row rank, and their convergence relies on the positivity of the observed distribution.

\section{Loss of potentials in hematopoiesis}\label{hemato}
Hematopoietic stem cells (HSC) are able to become progenitors that, in turn, may develop into mature blood cells. Understanding the process of forming mature blood cells, called hematopoiesis,  is one of the most important aims of cell biology, as it may help to develop new cancer treatments. The HSC progenitors can proliferate (produce cells of the same type) or differentiate (produce  cells of different types). Multiple experiments suggested that HSC progenitors are multipotent cells and differentiate by losing one of the potentials.  While the mature blood cells are unipotent, they do not proliferate or differentiate The differentiation is believed to be a hierarchical process, with HSC progenitors and mature blood cells at the highest and the lowest levels, respectively.

The models discussed below apply to the steady-state of hematopoiesis, under the assumption that cells neither proliferate nor die and can undergo only first phase of differentiation. Various hierarchical models for differentiation have been proposed \citep*[cf.][]{Kawamoto2010,YeCells}. The equal loss of potentials (ELP) model was introduced in \cite{Perie2014}, and is described next.  Denote by $MDB$ the three-potential HSC progenitor of the $M$, $D$, and $B$ mature blood cell types. During the first phase of differentiation, an $MDB$ progenitor can differentiate by losing either one or two potentials at the same time, and thus produce a cell of one of the six types: $M$, $D$, $B$, $MD$, $MB$, $DB$. 

Let $\bd p$ be the vector of probabilities of losing the corresponding potentials from $MDB$:
$$\bd p = (p_{*DB}, p_{M*B},p_{MD*},p_{**B},p_{*D*},p_{M**})\tr.$$
For example, $p_{*DB}$ is the probability of losing the  $M$ potential from $MDB$, $p_{M*B}$ is the probability of loosing the $D$ potential from $MDB$, and $p_{**B}$  is the probability of losing the $M$ and $D$ potentials from $MBD$ at the same time, and so on.
The ELP model assumes that ``the probability to lose two potentials at the same time is the product of
the probability of losing each of the potentials'' \citep[see Caption to Fig 3A,][]{Perie2014}:
\begin{equation}\label{Ploss}
p_{**B} = p_{*DB} \cdot  p_{M*B}, \qquad p_{M**} =  p_{MD*} \cdot   p_{M*B}, \qquad  p_{*D*} =  p_{*DB} \cdot   p_{MD*}.
\end{equation}
The model specified by (\ref{Ploss}) is the relational model generated by the matrix
\begin{equation}\label{mxELP}
\mathbf A = \begin{pmatrix}
   1& 0& 0& 1& 1& 0\\ 
   0& 1& 0& 1& 0& 1\\ 
   0& 0& 1& 0& 1& 1 
   \end{pmatrix},
   \end{equation}
   or, in a parametric form, 
\begin{align}\label{ELPour}
p_{*DB} &= {\alpha_M}, \quad p_{M*B} = {\alpha_D}, \quad p_{MD*} = {\alpha_B},\nonumber\\
p_{**B} &= {\alpha_M\alpha_D},\quad p_{*D*} = {\alpha_M\alpha_B},\quad p_{M**} = {\alpha_D\alpha_B},
\end{align}
where, using the notation in \cite{Perie2014}, $\alpha_M,\alpha_D, \alpha_B$ are the parameters associated with the loss of the corresponding potential from $MDB$. It can be easily verified that the relational model generated by (\ref{mxELP}) does not have the overall effect, so the normalization has to be added as a separate condition: 
$$Z = p_{*DB} + p_{M*B} + p_{MD*} + p_{**B} + p_{*D*} + p_{M**} = 1.$$
\cite{Perie2014} define the ELP model in the following parametric form:
\begin{align}\label{ELP}
p_{*DB} &= {\alpha_M}/{Z}, \quad p_{M*B} = {\alpha_D}/{Z}, \quad p_{MD*} = {\alpha_B}/{Z},\nonumber\\
p_{**B} &= {\alpha_M\alpha_D}/{Z},\quad p_{*D*} = {\alpha_M\alpha_B}/{Z},\quad p_{M**} = {\alpha_D\alpha_B}/{Z}.
\end{align}
That is, the authors rescaled the loss probabilities to force them sum to $1$.  In fact, (\ref{ELP}) is also a relational model; it is generated by 
\begin{equation}\label{ELPmatrix}
\bar{\mathbf{A}} = \left(\begin{array}{rrrrrr}
1&1&1&1&1&1\\
1& 0& 0& 1& 1& 0\\ 
0& 1& 0& 1& 0& 1\\ 
0& 0& 1& 0& 1& 1 
\end{array}\right),
\end{equation}
and can be obtained by adding the overall effect to the model defined by (\ref{mxELP}). Because, the original model does not have the overall effect, adding a row of $1$'s changed this model. One can check by substitution that the probabilities in (\ref{ELP}) do not satisfy the multiplicative constraints (\ref{Ploss}).  The estimates of the probabilities of loss of potentials from the $MDB$ cells shown in Figure 3B  of \cite{Perie2014}. In the notation used here,
\begin{eqnarray}\label{MLE1}
&& \hat{p}_{*DB} = 0.35, \quad \hat p_{M*B} = 0.08, \quad \hat{p}_{MD*} = 0.49, \nonumber \\
&&\hat p_{**B} =0.01, \quad \hat p_{*D*} = 0.06, \quad \hat p_{M**} = 0.01.
\end{eqnarray} 
These probabilities sum to $1$, but also do not satisfy  (\ref{Ploss}).

\section*{Acknowledgments}

The authors wish to thank Antonio Forcina for his thought-provoking discussions, Ingmar Glauche and Christoph Baldow for their help with understanding the main concepts of hematopoiesis, and  Wicher Bergsma. The second author is also a Recurrent Visiting Professor at the Central European University and the moral support received is acknowledged.

\bibliographystyle{apacite}

\bibliography{cells}

\begin{thebibliography}{}

\bibitem [\protect \citeauthoryear {%
Aitchison%
\ \BBA {} Silvey%
}{%
Aitchison%
\ \BBA {} Silvey%
}{%
{\protect \APACyear {1960}}%
}]{%
AitchSilvey60}
\APACinsertmetastar {%
AitchSilvey60}%
\begin{APACrefauthors}%
Aitchison, J.%
\BCBT {}\ \BBA {} Silvey, S\BPBI D.%
\end{APACrefauthors}%
\unskip\
\newblock
\APACrefYearMonthDay{1960}{}{}.
\newblock
{\BBOQ}\APACrefatitle {{M}aximum-likelihood estimation procedures and
  associated tests of significance} {{M}aximum-likelihood estimation procedures
  and associated tests of significance}.{\BBCQ}
\newblock
\APACjournalVolNumPages{J. Roy. Statist. Soc. Ser.B}{22}{}{154--171}.
\PrintBackRefs{\CurrentBib}

\bibitem [\protect \citeauthoryear {%
Andreas%
\ \BBA {} Klein%
}{%
Andreas%
\ \BBA {} Klein%
}{%
{\protect \APACyear {2015}}%
}]{%
AndreasKlein2015}
\APACinsertmetastar {%
AndreasKlein2015}%
\begin{APACrefauthors}%
Andreas, J.%
\BCBT {}\ \BBA {} Klein, D.%
\end{APACrefauthors}%
\unskip\
\newblock
\APACrefYearMonthDay{2015}{}{}.
\newblock
{\BBOQ}\APACrefatitle {When and why are log-linear models self-normalizing?}
  {When and why are log-linear models self-normalizing?}{\BBCQ}
\newblock
\BIn{} \APACrefbtitle {{Proceedings of the 2015 Conference of the North
  American Chapter of the Association for Computational Linguistics: Human
  Language Technologies}} {{Proceedings of the 2015 Conference of the North
  American Chapter of the Association for Computational Linguistics: Human
  Language Technologies}}\ (\BPGS\ 244--249).
\newblock
\APACaddressPublisher{}{ACM, New-York, USA}.
\PrintBackRefs{\CurrentBib}

\bibitem [\protect \citeauthoryear {%
Cox%
, Little%
\BCBL {}\ \BBA {} O'Shea%
}{%
Cox%
\ \protect \BOthers {.}}{%
{\protect \APACyear {2007}}%
}]{%
Cox}
\APACinsertmetastar {%
Cox}%
\begin{APACrefauthors}%
Cox, D\BPBI A.%
, Little, J.%
\BCBL {}\ \BBA {} O'Shea, D.%
\end{APACrefauthors}%
\unskip\
\newblock
\APACrefYear{2007}.
\newblock
\APACrefbtitle {{I}deals, varieties, and algorithms: an introduction to
  computational algebraic geometry and commutative algebra} {{I}deals,
  varieties, and algorithms: an introduction to computational algebraic
  geometry and commutative algebra}.
\newblock
\APACaddressPublisher{New York}{Springer}.
\PrintBackRefs{\CurrentBib}

\bibitem [\protect \citeauthoryear {%
Evans%
\ \BBA {} Forcina%
}{%
Evans%
\ \BBA {} Forcina%
}{%
{\protect \APACyear {2013}}%
}]{%
EvansForcina11}
\APACinsertmetastar {%
EvansForcina11}%
\begin{APACrefauthors}%
Evans, R\BPBI J.%
\BCBT {}\ \BBA {} Forcina, A.%
\end{APACrefauthors}%
\unskip\
\newblock
\APACrefYearMonthDay{2013}{}{}.
\newblock
{\BBOQ}\APACrefatitle {{T}wo algorithms for fitting constrained marginal
  models} {{T}wo algorithms for fitting constrained marginal models}.{\BBCQ}
\newblock
\APACjournalVolNumPages{Comput. Statist. Data Anal.}{66}{}{1--7}.
\PrintBackRefs{\CurrentBib}

\bibitem [\protect \citeauthoryear {%
Fienberg%
\ \BBA {} Rinaldo%
}{%
Fienberg%
\ \BBA {} Rinaldo%
}{%
{\protect \APACyear {2012}}%
}]{%
FienbergRinaldo2012}
\APACinsertmetastar {%
FienbergRinaldo2012}%
\begin{APACrefauthors}%
Fienberg, S\BPBI E.%
\BCBT {}\ \BBA {} Rinaldo, A.%
\end{APACrefauthors}%
\unskip\
\newblock
\APACrefYearMonthDay{2012}{}{}.
\newblock
{\BBOQ}\APACrefatitle {Maximum likelihood estimation in log-linear models}
  {Maximum likelihood estimation in log-linear models}.{\BBCQ}
\newblock
\APACjournalVolNumPages{Ann. Statist.}{40}{}{996--1023}.
\PrintBackRefs{\CurrentBib}

\bibitem [\protect \citeauthoryear {%
Forcina%
}{%
Forcina%
}{%
{\protect \APACyear {2017}}%
}]{%
Forcina2017}
\APACinsertmetastar {%
Forcina2017}%
\begin{APACrefauthors}%
Forcina, A.%
\end{APACrefauthors}%
\unskip\
\newblock
\APACrefYearMonthDay{2017}{}{}.
\newblock
\APACrefbtitle {Estimation for multiplicative models under multinomial
  sampling.} {Estimation for multiplicative models under multinomial sampling.}
\newblock
\APACaddressPublisher{}{arXiv:1704.06762}.
\PrintBackRefs{\CurrentBib}

\bibitem [\protect \citeauthoryear {%
Geiger%
, Meek%
\BCBL {}\ \BBA {} Sturmfels%
}{%
Geiger%
\ \protect \BOthers {.}}{%
{\protect \APACyear {2006}}%
}]{%
GeigerMeekSturm2006}
\APACinsertmetastar {%
GeigerMeekSturm2006}%
\begin{APACrefauthors}%
Geiger, D.%
, Meek, C.%
\BCBL {}\ \BBA {} Sturmfels, B.%
\end{APACrefauthors}%
\unskip\
\newblock
\APACrefYearMonthDay{2006}{}{}.
\newblock
{\BBOQ}\APACrefatitle {{O}n the toric algebra of graphical models} {{O}n the
  toric algebra of graphical models}.{\BBCQ}
\newblock
\APACjournalVolNumPages{Ann. Statist.}{34}{}{1463--1492}.
\PrintBackRefs{\CurrentBib}

\bibitem [\protect \citeauthoryear {%
Gr{\"u}nbaum%
}{%
Gr{\"u}nbaum%
}{%
{\protect \APACyear {2003}}%
}]{%
GrunbaumConvex}
\APACinsertmetastar {%
GrunbaumConvex}%
\begin{APACrefauthors}%
Gr{\"u}nbaum, B.%
\end{APACrefauthors}%
\unskip\
\newblock
\APACrefYear{2003}.
\newblock
\APACrefbtitle {Convex polytopes} {Convex polytopes}.
\newblock
\APACaddressPublisher{}{Springer}.
\PrintBackRefs{\CurrentBib}

\bibitem [\protect \citeauthoryear {%
H{\o}sgaard%
}{%
H{\o}sgaard%
}{%
{\protect \APACyear {2004}}%
}]{%
HosgaardCSImodels}
\APACinsertmetastar {%
HosgaardCSImodels}%
\begin{APACrefauthors}%
H{\o}sgaard, S.%
\end{APACrefauthors}%
\unskip\
\newblock
\APACrefYearMonthDay{2004}{}{}.
\newblock
{\BBOQ}\APACrefatitle {Statistical Inference in Context Specific Interaction
  Models for Contingency Tables} {Statistical inference in context specific
  interaction models for contingency tables}.{\BBCQ}
\newblock
\APACjournalVolNumPages{Scand. J. Statist.}{31}{}{143--158}.
\PrintBackRefs{\CurrentBib}

\bibitem [\protect \citeauthoryear {%
Kawamoto%
, Wada%
\BCBL {}\ \BBA {} Katsura%
}{%
Kawamoto%
\ \protect \BOthers {.}}{%
{\protect \APACyear {2010}}%
}]{%
Kawamoto2010}
\APACinsertmetastar {%
Kawamoto2010}%
\begin{APACrefauthors}%
Kawamoto, H.%
, Wada, H.%
\BCBL {}\ \BBA {} Katsura, Y.%
\end{APACrefauthors}%
\unskip\
\newblock
\APACrefYearMonthDay{2010}{}{}.
\newblock
{\BBOQ}\APACrefatitle {A revised scheme for developmental pathways of
  hematopoietic cells: the myeloid-based model} {A revised scheme for
  developmental pathways of hematopoietic cells: the myeloid-based
  model}.{\BBCQ}
\newblock
\APACjournalVolNumPages{International Immunology}{22}{}{65--70}.
\PrintBackRefs{\CurrentBib}

\bibitem [\protect \citeauthoryear {%
Kawamura%
, Matsuoka%
, Tajiri%
, Nishida%
\BCBL {}\ \BBA {} Hayashi%
}{%
Kawamura%
\ \protect \BOthers {.}}{%
{\protect \APACyear {1995}}%
}]{%
Kawamura1995}
\APACinsertmetastar {%
Kawamura1995}%
\begin{APACrefauthors}%
Kawamura, G.%
, Matsuoka, T.%
, Tajiri, T.%
, Nishida, M.%
\BCBL {}\ \BBA {} Hayashi, M.%
\end{APACrefauthors}%
\unskip\
\newblock
\APACrefYearMonthDay{1995}{}{}.
\newblock
{\BBOQ}\APACrefatitle {{E}ffectiveness of a sugarcane-fish combination as bait
  in trapping swimming crabs} {{E}ffectiveness of a sugarcane-fish combination
  as bait in trapping swimming crabs}.{\BBCQ}
\newblock
\APACjournalVolNumPages{Fisheries Research}{22}{}{155--160}.
\PrintBackRefs{\CurrentBib}

\bibitem [\protect \citeauthoryear {%
Klimova%
\ \BBA {} Rudas%
}{%
Klimova%
\ \BBA {} Rudas%
}{%
{\protect \APACyear {2012}}%
}]{%
KRbm}
\APACinsertmetastar {%
KRbm}%
\begin{APACrefauthors}%
Klimova, A.%
\BCBT {}\ \BBA {} Rudas, T.%
\end{APACrefauthors}%
\unskip\
\newblock
\APACrefYearMonthDay{2012}{}{}.
\newblock
{\BBOQ}\APACrefatitle {{C}oordinate free analysis of trends in {British} social
  mobility} {{C}oordinate free analysis of trends in {British} social
  mobility}.{\BBCQ}
\newblock
\APACjournalVolNumPages{J. Appl. Stat.}{39}{}{1681--1691}.
\PrintBackRefs{\CurrentBib}

\bibitem [\protect \citeauthoryear {%
Klimova%
\ \BBA {} Rudas%
}{%
Klimova%
\ \BBA {} Rudas%
}{%
{\protect \APACyear {2014}}%
}]{%
gIPFpackage}
\APACinsertmetastar {%
gIPFpackage}%
\begin{APACrefauthors}%
Klimova, A.%
\BCBT {}\ \BBA {} Rudas, T.%
\end{APACrefauthors}%
\unskip\
\newblock
\APACrefYearMonthDay{2014}{}{}.
\newblock
{\BBOQ}\APACrefatitle {{gIPFrm}: {G}eneralized {I}terative {P}roportional
  {F}itting for {R}elational {M}odels} {{gIPFrm}: {G}eneralized {I}terative
  {P}roportional {F}itting for {R}elational {M}odels}{\BBCQ}\
  [\bibcomputersoftwaremanual].
\newblock
\begin{APACrefURL}
  \url{{http://cran.r-project.org/web/packages/gIPFrm/index.html}}
  \end{APACrefURL}
\newblock
\APACrefnote{{accessed on June 9, 2017. R package version 2.0}}
\PrintBackRefs{\CurrentBib}

\bibitem [\protect \citeauthoryear {%
Klimova%
\ \BBA {} Rudas%
}{%
Klimova%
\ \BBA {} Rudas%
}{%
{\protect \APACyear {2015}}%
}]{%
KRipf1}
\APACinsertmetastar {%
KRipf1}%
\begin{APACrefauthors}%
Klimova, A.%
\BCBT {}\ \BBA {} Rudas, T.%
\end{APACrefauthors}%
\unskip\
\newblock
\APACrefYearMonthDay{2015}{}{}.
\newblock
{\BBOQ}\APACrefatitle {{I}terative scaling in curved exponential families}
  {{I}terative scaling in curved exponential families}.{\BBCQ}
\newblock
\APACjournalVolNumPages{Scand. J. Statist.}{42}{}{832--847}.
\PrintBackRefs{\CurrentBib}

\bibitem [\protect \citeauthoryear {%
Klimova%
\ \BBA {} Rudas%
}{%
Klimova%
\ \BBA {} Rudas%
}{%
{\protect \APACyear {2016}}%
}]{%
KRextended}
\APACinsertmetastar {%
KRextended}%
\begin{APACrefauthors}%
Klimova, A.%
\BCBT {}\ \BBA {} Rudas, T.%
\end{APACrefauthors}%
\unskip\
\newblock
\APACrefYearMonthDay{2016}{}{}.
\newblock
{\BBOQ}\APACrefatitle {On the closure of relational models} {On the closure of
  relational models}.{\BBCQ}
\newblock
\APACjournalVolNumPages{J. Multivariate Anal.}{143}{}{440--452}.
\PrintBackRefs{\CurrentBib}

\bibitem [\protect \citeauthoryear {%
Klimova%
, Rudas%
\BCBL {}\ \BBA {} Dobra%
}{%
Klimova%
\ \protect \BOthers {.}}{%
{\protect \APACyear {2012}}%
}]{%
KRD11}
\APACinsertmetastar {%
KRD11}%
\begin{APACrefauthors}%
Klimova, A.%
, Rudas, T.%
\BCBL {}\ \BBA {} Dobra, A.%
\end{APACrefauthors}%
\unskip\
\newblock
\APACrefYearMonthDay{2012}{}{}.
\newblock
{\BBOQ}\APACrefatitle {{R}elational models for contingency tables}
  {{R}elational models for contingency tables}.{\BBCQ}
\newblock
\APACjournalVolNumPages{J. Multivariate Anal.}{104}{}{159--173}.
\PrintBackRefs{\CurrentBib}

\bibitem [\protect \citeauthoryear {%
Koller%
\ \BBA {} Friedman%
}{%
Koller%
\ \BBA {} Friedman%
}{%
{\protect \APACyear {2009}}%
}]{%
ProbGraphM}
\APACinsertmetastar {%
ProbGraphM}%
\begin{APACrefauthors}%
Koller, D.%
\BCBT {}\ \BBA {} Friedman, N.%
\end{APACrefauthors}%
\unskip\
\newblock
\APACrefYear{2009}.
\newblock
\APACrefbtitle {Probabilistic Graphical Models: Principles and Techniques}
  {Probabilistic graphical models: Principles and techniques}.
\newblock
\APACaddressPublisher{}{Chapman $\&$ Hall}.
\PrintBackRefs{\CurrentBib}

\bibitem [\protect \citeauthoryear {%
Mnih%
\ \BBA {} Teh%
}{%
Mnih%
\ \BBA {} Teh%
}{%
{\protect \APACyear {2012}}%
}]{%
MnihTeh2012}
\APACinsertmetastar {%
MnihTeh2012}%
\begin{APACrefauthors}%
Mnih, A.%
\BCBT {}\ \BBA {} Teh, Y\BPBI W.%
\end{APACrefauthors}%
\unskip\
\newblock
\APACrefYearMonthDay{2012}{}{}.
\newblock
{\BBOQ}\APACrefatitle {A fast and simple algorithm for training neural
  probabilistic language models} {A fast and simple algorithm for training
  neural probabilistic language models}.{\BBCQ}
\newblock
\BIn{} J.~Langford\ \BBA {} J.~Pineau\ (\BEDS), \APACrefbtitle {{Proceedings of
  the 29th International Conference on Machine learning (ICML 2012), Edinburgh,
  Scotland, UK}} {{Proceedings of the 29th International Conference on Machine
  learning (ICML 2012), Edinburgh, Scotland, UK}}\ (\BPG~1751-1758).
\newblock
\APACaddressPublisher{}{Omnipress}.
\PrintBackRefs{\CurrentBib}

\bibitem [\protect \citeauthoryear {%
Nyman%
, Pensar%
, Koski%
\BCBL {}\ \BBA {} Corander%
}{%
Nyman%
\ \protect \BOthers {.}}{%
{\protect \APACyear {2016}}%
}]{%
NymanCSI}
\APACinsertmetastar {%
NymanCSI}%
\begin{APACrefauthors}%
Nyman, H.%
, Pensar, J.%
, Koski, T.%
\BCBL {}\ \BBA {} Corander, J.%
\end{APACrefauthors}%
\unskip\
\newblock
\APACrefYearMonthDay{2016}{}{}.
\newblock
{\BBOQ}\APACrefatitle {Context-specific independence in graphical log-linear
  models} {Context-specific independence in graphical log-linear
  models}.{\BBCQ}
\newblock
\APACjournalVolNumPages{Computational Statistics}{31}{}{1493--1512}.
\PrintBackRefs{\CurrentBib}

\bibitem [\protect \citeauthoryear {%
Peri\'{e}%
\ \protect \BOthers {.}}{%
Peri\'{e}%
\ \protect \BOthers {.}}{%
{\protect \APACyear {2014}}%
}]{%
Perie2014}
\APACinsertmetastar {%
Perie2014}%
\begin{APACrefauthors}%
Peri\'{e}, L.%
, Hodgkin, P\BPBI D.%
, Naik, S\BPBI H.%
, Schumacher, T\BPBI N.%
, {de Boer}, R\BPBI J.%
\BCBL {}\ \BBA {} Duffy, K\BPBI R.%
\end{APACrefauthors}%
\unskip\
\newblock
\APACrefYearMonthDay{2014}{}{}.
\newblock
{\BBOQ}\APACrefatitle {Determining lineage pathways from cellular barcoding
  experiements} {Determining lineage pathways from cellular barcoding
  experiements}.{\BBCQ}
\newblock
\APACjournalVolNumPages{{Cell Reports}}{6}{}{617--624}.
\PrintBackRefs{\CurrentBib}

\bibitem [\protect \citeauthoryear {%
Sturmfels%
}{%
Sturmfels%
}{%
{\protect \APACyear {1996}}%
}]{%
SturBook}
\APACinsertmetastar {%
SturBook}%
\begin{APACrefauthors}%
Sturmfels, B.%
\end{APACrefauthors}%
\unskip\
\newblock
\APACrefYear{1996}.
\newblock
\APACrefbtitle {{G}r{\"o}bner bases and convex polytopes} {{G}r{\"o}bner bases
  and convex polytopes}.
\newblock
\APACaddressPublisher{Providence RI}{AMS}.
\PrintBackRefs{\CurrentBib}

\bibitem [\protect \citeauthoryear {%
Ye%
, Huang%
\BCBL {}\ \BBA {} Guo%
}{%
Ye%
\ \protect \BOthers {.}}{%
{\protect \APACyear {2017}}%
}]{%
YeCells}
\APACinsertmetastar {%
YeCells}%
\begin{APACrefauthors}%
Ye, F.%
, Huang, W.%
\BCBL {}\ \BBA {} Guo, G.%
\end{APACrefauthors}%
\unskip\
\newblock
\APACrefYearMonthDay{2017}{}{}.
\newblock
{\BBOQ}\APACrefatitle {Studying hematopoiesis using single-cell technologies}
  {Studying hematopoiesis using single-cell technologies}.{\BBCQ}
\newblock
\APACjournalVolNumPages{Journal of Hematology \& Oncology}{10}{}{}.
\PrintBackRefs{\CurrentBib}

\end{thebibliography}

\end{document}